\documentclass[11pt,pdftex,letterpaper]{article}
\pdfoutput=1
\usepackage{float}
\floatstyle{boxed} 
\restylefloat{figure}
\usepackage{subfig}

\usepackage{verbatim}
\usepackage{amssymb}
\usepackage{lipsum}
\usepackage{multicol}
\usepackage{amsmath}
\usepackage{framed}
\usepackage{amsthm}
\usepackage{enumitem}
\setlist{nolistsep}
\usepackage{tikz}
\usepackage[margin=1.0in]{geometry}
\usepackage{color}
\usepackage{cite}
\newif\ifcode
\codefalse
\codetrue
\usepackage{macros}
\ifcode
\usepackage[ruled]{algorithm}
\usepackage{algpseudocode}
\usepackage{ioa_code}

\newtheorem{theorem}{Theorem}

\newtheorem{claim}[theorem]{Claim}

\newtheorem{corollary}[theorem]{Corollary}
\newtheorem{definition}{Definition}
\newtheorem{lemma}[theorem]{Lemma}

\newcounter{linenumber}

\newcommand{\remove}[1]{}

\newcommand{\Wset}{\textit{Wset}}
\newcommand{\Rset}{\textit{Rset}}
\newcommand{\Dset}{\textit{Dset}}

\newcommand{\txns}{\textit{txns}}

\newcommand{\Read}{\textit{read}}
\newcommand{\Write}{\textit{write}}
\newcommand{\TryC}{\textit{tryC}}

\newcommand{\ignore}[1]{}

\begin{document}
\bibliographystyle{abbrv}

\title{On Partial Wait-Freedom in Transactional Memory}

\author{
Petr Kuznetsov$^1$~~~Srivatsan Ravi$^2$ \\
$^1$\normalsize T\'el\'ecom ParisTech \\
$^2$\normalsize TU Berlin
}

\date{}
\maketitle
\thispagestyle{empty}
\begin{abstract}
Transactional memory (TM) is a convenient synchronization tool that allows
concurrent threads to declare sequences of instructions on shared data as
speculative \emph{transactions} with ``all-or-nothing'' semantics. 
It is known that dynamic transactional memory cannot provide
\emph{wait-free} progress in the sense that every transaction
commits in a finite number of its own steps. 
In this paper, we explore the costs of providing
wait-freedom to only a \emph{subset} of transactions.  
Since most transactional workloads are believed to be read-dominated, we require that read-only
transactions commit in the wait-free manner, while updating transactions 
are guaranteed to commit only if they run in the absence of concurrency.
We show that this kind of partial wait-freedom, combined
with attractive requirements like read invisibility or disjoint-access
parallelism, incurs considerable complexity costs.
\end{abstract}

\newpage
\pagenumbering{arabic}\setcounter{page}{1}
%
\section{Introduction}
\label{sec:intro}
The transactional memory abstraction (\emph{TM}) allows concurrent
processes to declare  sequences
of 
operations on shared data as \emph{atomic transaction}. 
A transaction may \emph{commit} 
in which case its updates to shared data items ``take effect'', or \emph{abort}, in which
case the transaction does not affect other transactions.
A TM \emph{implementation} provides processes with algorithms for
transactional operations (\emph{read}, \emph{write}, \emph{tryCommit}) using low-level \emph{base objects}.

A natural consistency criterion provided by most TM implementations that we assume in this paper is
\emph{strict serializability}: 
all committed transactions appear to execute sequentially in some
total order respecting the timing of non-overlapping transactions.
This is strictly weaker than conditions like \emph{opacity}~\cite{tm-book} and \emph{virtual-world consistency}~\cite{damien-vw}, 
which require \emph{every} transaction (including aborted and incomplete ones) to observe a consistent state.

The spectrum of progress properties specifying the conditions under
which a transaction must commit appears more interesting.
Perhaps, the most attractive progress property a TM may satisfy is
\emph{wait-freedom}:
every transaction commits regardless of the behavior of concurrent
processes~\cite{Her91}.
Intuitively, wait-free progress is desirable because it guarantees
that a transaction takes effect in a finite
number of its own steps and, thus, the progress of the corresponding
process is not affected by other processes.  
It is easy to see, however, that dynamic TMs in which
data sets of transactions are not known in advance do not allow for
wait-free implementations~\cite{tm-book}. 
Suppose that a transaction
$T_1$ reads data item $X$, then a concurrent transaction $T_2$ reads
data item $Y$, writes to $X$ and commits, and finally $T_2$ writes to
$Y$. 
Since $T_1$ has read the ``old'' value in $X$ and $T_2$ has read the
``old'' value  in $Y$, there is no way to commit $T_1$ and order the
two transactions in a sequential execution.   
As this scenario can be repeated arbitrarily often, even the weaker guarantee of \emph{local progress} that only requires that
each transaction \emph{eventually} commits if repeated
sufficiently often, cannot be ensured by \emph{any} strictly
serializable TM implementation,
regardless of the base objects it uses~\cite{bushkov2012}.           
 
But can we ensure that at least \emph{some} transactions commit wait-free?     
%
It is often argued that many realistic workloads are
\emph{read-dominated}: the proportion of read-only transactions is
higher than that of updating ones, or read-only transactions have much
larger data sets than updating ones~\cite{stmbench7, Attiya09-tmread}. 
Therefore, it seems natural to require first that read-only transactions
commit wait-free. Of course, the progress guarantees for updating transaction have to
be weaker.
Since this paper focuses on complexity lower bounds for
read-only transactions, we consider a very weak property saying
that an updating transaction is guaranteed to commit only if it runs in the absence of concurrency (we refer to this
as sequential progress for updating transactions). 
  
First, we show that if we wish to derive implementations
that use \emph{invisible reads},
i. e., assuming that read-only transactions do not
apply any nontrivial primitives on the base memory, then the implementation must maintain unbounded sets of values
for every data item.
Since read-invisibility is believed important for
(most common) read-dominated workloads, our result suggests that these implementations 
may not be practical due to their space complexity. 

We then focus on \emph{disjoint-access-parallel} (DAP)
implementations~\cite{israeli-disjoint,AHM09}.
The idea of DAP is to allow transactions that do not contend on the
same data item to proceed independently of each other without memory contention. 
A \emph{strict DAP} TM implementation ensures that two transactions
\emph{contend} on a base object (i.e. both access the base object and at least one modifies it)
only if they access a common data item~\cite{GK08-OFT,DStransaction06}.
Interestingly, we prove that it is impossible to implement strict DAP implementations that ensure wait-free progress for
read-only transactions and sequential progress for updating transactions. 
Thus, two transactions that access mutually disjoint data sets may prevent each other from committing.

A less restrictive definition of DAP satisfied by several popular TM implementations~\cite{nztm,fraser,HLM+03}
is \emph{weak DAP}~\cite{AHM09} which is typically defined using a  
\emph{conflict graph} defined for each pair of concurrent transactions $T_1$ and $T_2$.
The vertices of the graph are data items accessed by $T_1$, $T_2$ and all 
transactions that are concurrent to them.  There is an edge between
two data items if the item is accessed by two concurrent transactions
in the set.  
A weak DAP implementation ensures that transactions $T_1$ and $T_2$ 
are allowed to \emph{concurrently contend}  
on a base object (i.e., to concurrently have enabled operations
on it one of which is about to modify the object) only if there
is a path in their conflict graph between a data item accessed by $T_1$ and a data item accessed by $T_2$.
For weak DAP TMs, we show that a read-only transaction (with possibly unbounded read
set) must sometimes perform at least one
expensive synchronization pattern~\cite{AGK11-popl} per read
operation
(the expensive-pattern complexity of a read-only transaction is linear in
the size of its data set).     
These patterns include  \emph{read-after-write} (or \emph{RAW}), which
incurs a costly memory fence on most CPU architectures, or  
a  \emph{atomic-write-after-read} (or \emph{AWAR}), typically
instantiated as atomic \emph{compare-and-swap}.
The metric appears to be a more adequate complexity measure than
simple step complexity, as it accounts for expensive cache-coherence
operations or conditional instructions. 

Overall, our results highlight considerable complexity costs incurred
by requiring partial wait-freedom for read-only transactions, even
when remaining updating transactions are only provided with extremely weak
progress. We hope this paper provides a better understanding of the
pros and contras of diversified progress guarantees for different workloads.   

\vspace{2mm}\noindent\textbf{Roadmap.}
Section~\ref{sec:model} describes our TM model and Section~\ref{sec:classes} defines the TM classes considered in this paper.
We present in Section~\ref{sec:rwl2}, the space complexity of implementations that use invisible reads.
In Section~\ref{sec:lpdap}, we prove the impossibility result concerning strict DAP TMs and in Section~\ref{sec:rwl}, 
the lower bound on the
number of expensive synchronization patterns for weak DAP TMs.
Section~\ref{sec:related} relates our work to earlier results. 
In Section~\ref{sec:disc}, we conclude the paper and discuss open questions.
%

%
\section{Model}
\label{sec:model}
\vspace{1mm}\noindent\textbf{TM interface.}
A \emph{transactional memory} (in short, \emph{TM})
supports \emph{transactions} for
reading and writing on a finite set $\mathcal{X}$ of data items,
referred to as \emph{t-objects}.
Every transaction $T_k$ has a unique identifier $k$. We assume no bound on the size of a t-object i.e. 
the number of possible different values a t-object can have.
A transaction $T_k$ may contain the following \emph{t-operations},
each being a matching pair of an \emph{invocation} and a \emph{response}:
$\Read_k(X)$ returns a value in some domain $V$ (denoted $\Read_k(X) \rightarrow v$)
or a special value $A_k\notin V$ (\emph{abort});
$\Write_k(X,v)$, for a value $v \in V$,
returns \textit{ok} or $A_k$;
$\TryC_k$ returns $C_k\notin V$ (\emph{commit}) or $A_k$.

\vspace{2mm}\noindent\textbf{TM implementations.}
We consider an asynchronous shared-memory system in which
a set of $n$ processes, communicate by applying \emph{primitives} on shared \emph{base objects}.
We assume that processes issue transactions sequentially i.e. a process starts a new transaction
only after the previous transaction has committed or aborted.
A TM \emph{implementation} provides processes with algorithms
for implementing $\Read_k$, $\Write_k$ and $\TryC_k()$
of a transaction $T_k$ by \emph{applying} \emph{primitives} from a set of shared \emph{base objects}, each of which is 
assigned an \emph{initial value}.
We assume that these primitives are \emph{deterministic}.
A primitive is a generic \emph{read-modify-write} (\emph{RMW}) procedure applied to a base object~\cite{G05,Her91}.
It is characterized by a pair of functions $\langle g,h \rangle$:
given the current state of the base object, $g$ is an \emph{update function} that
computes its state after the primitive is applied, while $h$ 
is a \emph{response function} that specifies the outcome of the primitive returned to the process.
A RMW primitive is \emph{trivial} if it never changes the value of the base object to which it is applied.
Otherwise, it is \emph{nontrivial}.

\vspace{2mm}\noindent\textbf{Executions and configurations.}
An \emph{event} of a transaction $T_k$ (sometimes we say \emph{step} of $T_k$)
is an invocation or response of a t-operation performed by $T_k$ or a 
RMW primitive $\langle g,h \rangle$ applied by $T_k$ to a base object $b$
along with its response $r$ (we call it a \emph{RMW event} and write $(b, \langle g,h\rangle, r,k)$).

A \emph{configuration} (of a TM implementation) specifies the value of each base object and 
the state of each process.
The \emph{initial configuration} is the configuration in which all 
base objects have their initial values and all processes are in their initial states.

An \emph{execution fragment} is a (finite or infinite) sequence of events.
An \emph{execution} of a TM implementation $M$ is an execution
fragment where, starting from the initial configuration, each event is
issued according to $M$ and each response of a RMW event $(b, \langle
g,h\rangle, r,k)$ matches the state of $b$ resulting from all
preceding events.
An execution $E\cdot E'$, denoting the concatenation of $E$ and $E'$,
is an \emph{extension} of $E$ and we say that $E'$ \emph{extends} $E$.

Let $E$ be an execution fragment.
For every transaction identifier $k$,
$E|k$ denotes the subsequence of $E$ restricted to events of
transaction $T_k$.
If $E|k$ is non-empty,
we say that $T_k$ \emph{participates} in $E$, else we say $E$ is \emph{$T_k$-free}.
Two executions $E$ and $E'$ are \emph{indistinguishable} to a set $\mathcal{T}$ of transactions, if
for each transaction $T_k \in \mathcal{T}$, $E|k=E'|k$.

The \emph{read set} (resp., the \emph{write set}) of a transaction $T_k$ in an execution $E$,
denoted $\Rset(T_k)$ (and resp. $\Wset(T_k)$), is the set of t-objects on which $T_k$ invokes reads (and resp. writes) in $E$.
The \emph{data set} of $T_k$ is $\Dset(T_k)=\Rset(T_k)\cup\Wset(T_k)$.
A transaction is called \emph{read-only} if $\Wset(T_k)=\emptyset$; \emph{write-only} if $\Rset(T_k)=\emptyset$ and
\emph{updating} if $\Wset(T_k)\neq\emptyset$.

\vspace{2mm}\noindent\textbf{Transaction orders.}
Let $\txns(E)$ denote the set of transactions that participate in $E$.
The \emph{history exported by an execution $E$}
is the subsequence of $E$ consisting of the invocation and response events of t-operations.
Two histories $H$ and $H'$ are \emph{equivalent} if $\txns(H) = \txns(H')$
and for every transaction $T_k \in \txns(H)$, $H|k=H'|k$.
An execution $E$ is \emph{sequential} if every invocation of
a t-operation is either the last event in the history $H$ exported by $E$ or
is immediately followed by a matching response.
We assume that executions are \emph{well-formed} i.e. for all $T_k$, $E|k$ is
sequential and has no events after $A_k$ or $C_k$.
A transaction $T_k\in \txns(E)$ is \emph{complete in $E$} if
$E|k$ ends with a response event.
The execution $E$ is \emph{complete} if all transactions in $\txns(E)$
are complete in $E$.
A transaction $T_k\in \txns(E)$ is \emph{t-complete} if $E|k$
ends with $A_k$ or $C_k$; otherwise, $T_k$ is \emph{t-incomplete}.
$T_k$ is \emph{committed} (resp., \emph{aborted}) in $E$
if the last event of $T_k$ is $C_k$ (resp., $A_k$).
The execution $E$ is \emph{t-complete} if all transactions in
$\txns(E)$ are t-complete.

For transactions $\{T_k,T_m\} \in \txns(E)$, we say that $T_k$ \emph{precedes}
$T_m$ in the \emph{real-time order} of $E$, denoted $T_k\prec_E^{RT} T_m$,
if $T_k$ is t-complete in $E$ and
the last event of $T_k$ precedes the first event of $T_m$ in $E$.
If neither $T_k\prec_E^{RT} T_m$ nor $T_m\prec_E^{RT} T_k$,
then $T_k$ and $T_m$ are \emph{concurrent} in $E$.
An execution $E$ is \emph{t-sequential} if there are no concurrent
transactions in $E$.

\vspace{2mm}\noindent\textbf{Contention.}
We say that a configuration $C$ after an execution $E$ is \emph{quiescent} (and resp. \emph{t-quiescent})
if every transaction $T_k \in \ms{txns}(E)$ is complete (and resp. t-complete) in $C$.
If a transaction $T$ is incomplete in an execution $E$, it has exactly one \emph{enabled} event, 
which is the next event the transaction will perform according to the TM implementation.
Events $e$ and $e'$ of an execution $E$  \emph{contend} on a base
object $b$ if they are both events on $b$ in $E$ and at least 
one of them is nontrivial (the event is trivial if it is the application of a trivial primitive; otherwise,
nontrivial).

We say that a transaction $T$ is \emph{poised to apply an event $e$ after $E$} 
if $e$ is the next enabled event for $T$ in $E$.
We say that transactions $T$ and $T'$ \emph{concurrently contend on $b$ in $E$} 
if they are each poised to apply contending events on $b$ after $E$.

We say that an execution fragment $E$ is \emph{step contention-free for t-operation $op_k$} if the events of $E|op_k$ 
are contiguous in $E$.
We say that an execution fragment $E$ is \emph{step contention-free for $T_k$} if the events of $E|k$ are contiguous in $E$.
We say that $E$ is \emph{step contention-free} if $E$ is step contention-free for all transactions that participate in $E$.
\section{TM classes}
\label{sec:classes}
\vspace{2mm}\noindent\textbf{TM correctness.}
For a history $H$, a \emph{completion of $H$}, denoted ${\bar H}$,
is a history derived from $H$ through the following procedure:
(1) for every incomplete t-operation $op_k$ of $T_k \in \txns(H)$ in $H$,
if $op_k=\Read_k \vee \Write_k$, 
insert $A_k$ somewhere after the invocation of $op_k$; 
otherwise, if $op_k=\TryC_k()$, 
insert $C_k$ or $A_k$ somewhere after the last event of $T_k$.
(2) for every complete transaction $T_k$ that is not t-complete, insert $\TryC_k\cdot A_k$ somewhere after the 
last event of transaction $T_k$.

For simplicity of presentation, we assume that each execution $E$
begins with an ``imaginary'' transaction that writes initial
values to all t-objects and commits before any other transaction
begins in $E$.
Let $E$ be a
t-sequential execution.
For every operation $\Read_k(X)$ in $E$,
we define the \emph{latest written value} of $X$ as follows:
\begin{enumerate}
\item[(1)] 
If $T_k$ contains a $\Write_k(X,v)$ preceding $\Read_k(X)$,
then the latest written value of $X$ is the value of the latest such write to $X$.
\item[(2)] 
Otherwise, if $E$ contains a $\Write_m(X,v)$,
$T_m$ precedes $T_k$, and $T_m$ commits in $E$,
then the latest written value of $X$ is the value
of the latest such write to $X$ in $E$.
(This write is well-defined since $E$ starts with an initial transaction writing to
all t-objects.)
\end{enumerate}
We say that $\Read_k(X)$ is \emph{legal} in a t-sequential execution $E$ if it returns the
\emph{latest written value} of $X$ in $E$, and $E$ is \emph{legal}
if every $\Read_k(X)$ in $E$ that does not return $A_k$ is legal in $E$.
\begin{definition}
(Strict serializability) A finite history $H$ is \emph{strictly serializable} if there
is a legal t-complete t-sequential history $S$,
such that
(1) for any two transactions $T_k,T_m \in \txns(H)$,
if $T_k \prec_H^{RT} T_m$, then $T_k$ precedes $T_m$ in $S$, and
(2) $S$ is equivalent to $\ms{cseq}(\bar H)$, where $\bar H$ is some
completion of $H$ and $\ms{cseq}(\bar H)$ is the subsequence of $\bar H$ reduced to
committed transactions in $\bar H$.

We refer to $S$ as a \emph{serialization} of $H$.
\end{definition}
\begin{definition}
\label{def:rwf} 
(The class $\mathcal{RWF}$) 
A TM implementation $M \in \mathcal{RWF}$ \emph{iff} it is strictly serializable and
in its every execution:
\begin{itemize}
\item 
(\emph{wait-free progress for read-only transactions}) 
every read-only transaction 
commits
in a finite number of its steps, and
\item 
(\emph{sequential progress for updating transactions}) 
every transaction running step contention-free from a t-quiescent
configuration, commits in a finite number of its steps.
\end{itemize}
\end{definition}
\vspace{2mm}\noindent\textbf{Invisible reads.}
We say that a TM implementation $M$ uses \emph{invisible reads} if in every execution $E$ of $M$, and
every read-only transaction $T_k\in \ms{txns}(E)$, $E|k$ does not contain any nontrivial events.

\vspace{2mm}\noindent\textbf{Disjoint-access parallelism (DAP).}
A TM implementation $M$ is \emph{strictly disjoint-access parallel
  (strict DAP)} if, for
all executions $E$ of $M$, and for all transactions $T_i$ and $T_j$ that participate in $E$, 
$T_i$ and $T_j$ contend on a base object in $E$ only if 
$\Dset(T_i)\cap \Dset(T_j)\neq \emptyset$~\cite{tm-book}.

For an execution fragment, let  $\tau_{E}(T_1,T2)$  denote the set of transactions ($T_1$ and $T_2$ included)
that are concurrent to at least one of $T_1$ and $T_2$ in $E$.
Let $G(T_1,T_2,E)$ be an undirected graph with the vertex set $\cup_{T \in \tau_{E}(T_1,T_2)} \Dset(T)$
and there is an edge
between t-objects $X$ and $Y$ \emph{iff} there exists $T \in \tau_{E}(T_1,T_2)$ such that 
$\{X,Y\} \in \Dset(T)$.
We say that $T_1$ and $T_2$ are \emph{disjoint-access} in $E$
if there is no path between a t-object in $\Dset(T_1)$ and a t-object in $\Dset(T_2)$ in $G(T_1,T_2,E)$.
A TM implementation $M$ is \emph{weak DAP} if, in
all executions $E$ of $M$, 
any two transactions $T_1$ and $T_2$ 
concurrently contend on the same base object after $E$ only if   
$T_1$ and $T_2$ are not disjoint-access in $E$ or there exists a t-object $X \in \Dset(T_1) \cap \Dset(T_2)$~\cite{AHM09,PFK10}.

Observe that every strict DAP TM implementation satisfies weak DAP, but not vice-versa.

We first prove the
following auxiliary result, inspired by \cite{AHM09}:
\begin{lemma}
\label{lm:dap}
Let $M$ be any weak DAP TM implementation.
Let $\alpha\cdot \rho_1 \cdot \rho_2$ be any execution of $M$ where
$\rho_1$ (and resp. $\rho_2$) is the step contention-free
execution fragment of transaction $T_1 \not\in \ms{txns}(\alpha)$ (and resp. $T_2 \not\in \ms{txns}(\alpha)$) 
and transactions $T_1$, $T_2$ are disjoint-access in $\alpha\cdot \rho_1 \cdot \rho_2$. 
Then, $T_1$ and $T_2$ do not contend on any base object in $\alpha\cdot \rho_1 \cdot \rho_2$.
\end{lemma}
\begin{proof}
Suppose, by contradiction that $T_1$ and $T_2$ contend on the same base object in $\alpha\cdot \rho_1\cdot \rho_2$.

If in $\rho_1$, $T_1$ performs a nontrivial event on a base object on which they contend, let $e_1$ be the last
event in $\rho_1$ in which $T_1$ performs such an event to some base object $b$ and $e_2$, the first event
in $\rho_2$ that accesses $b$.
Otherwise, $T_1$
only performs trivial events in $\rho_1$ to base objects on which it contends with $T_2$ in $\alpha\cdot \rho_1\cdot \rho_2$:
let $e_2$ be the first event in $\rho_2$ in which $\rho_2$ performs a nontrivial event to some base object $b$
on which they contend and $e_1$, the last event of $\rho_1$ in $T_1$ that accesses $b$.

Let $\rho_1'$ (and resp. $\rho_2'$) be the longest prefix of $\rho_1$ (and resp. $\rho_2$) that does not include
$e_1$ (and resp. $e_2$).
Since before accessing $b$, the execution is step contention-free for $T_1$, $\alpha \cdot
\rho_1'\cdot \rho_2'$ is an execution of $M$.
By construction, $T_1$ and $T_2$ are disjoint-access in $\alpha \cdot \rho_1'\cdot \rho_2'$
and $\alpha\cdot \rho_1 \cdot \rho_2'$ is indistinguishable to $p_2$ from $\alpha\cdot \rho_1' \cdot \rho_2'$.
Hence, $T_1$ and
$T_2$ are poised to apply contending events $e_1$ and $e_2$ on $b$ in the configuration after 
$\alpha\cdot \rho_1' \cdot \rho_2'$---a contradiction since $T_1$ and $T_2$ cannot concurrently contend on the same base object.   
\end{proof}
\section{On the cost of invisible reads}
\label{sec:rwl2}
We prove that every TM implementation $M\in \mathcal{RWF}$ that uses invisible reads 
must keep unbounded sets of values for every t-object.
To do so, for every $c \in \mathbb{N}$, we construct an execution of $M$ that \emph{maintains at least $c$ distinct values
for every t-object}. We require the following technical definition:
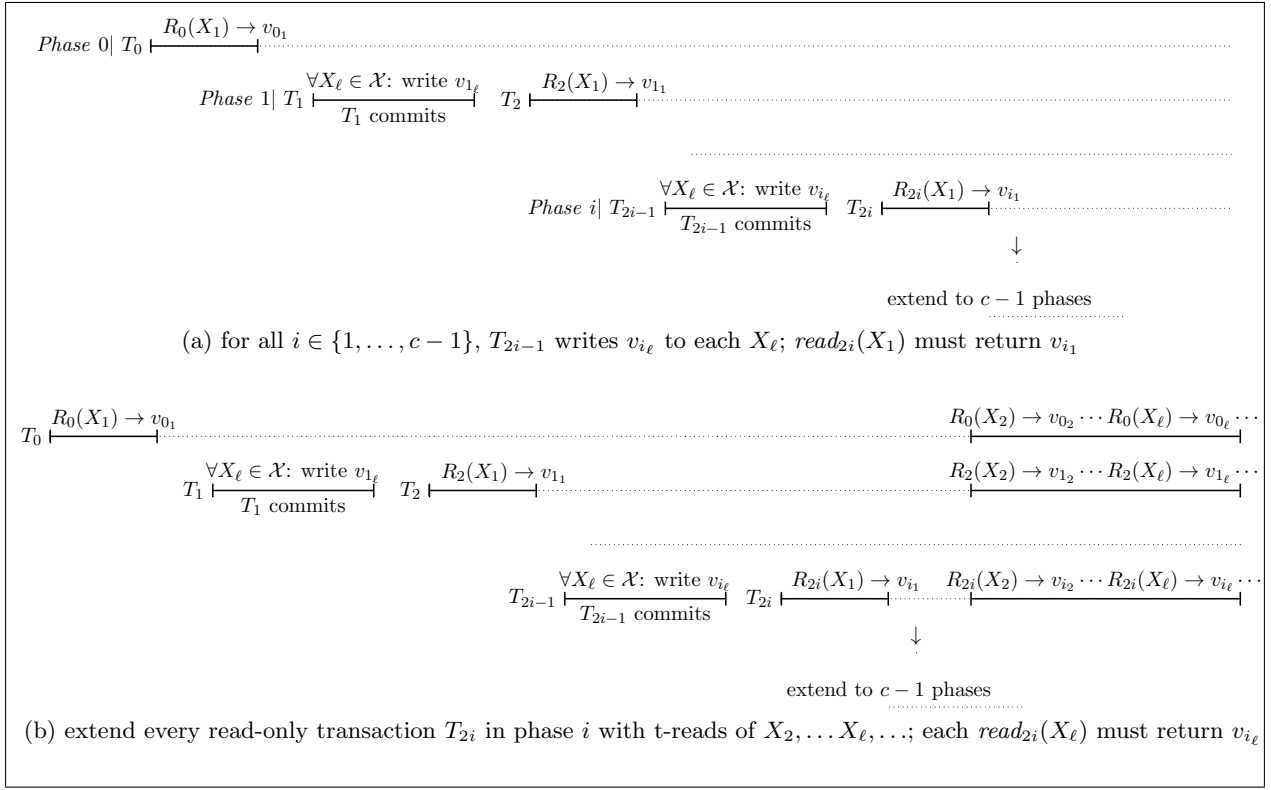
\begin{figure*}[t]
\begin{center}
	\subfloat[for all $i\in \{1,\ldots , c-1 \}$, $T_{2i-1}$ writes $v_{i_{\ell}}$ to each $X_{\ell}$; $\Read_{2i}(X_1)$ must return $v_{i_{1}}$ \label{sfig:res-0}]{\scalebox{0.72}[0.72]{\begin{tikzpicture}
\node (r1) at (1.4,0) [] {};
\node (r2) at (8.4,-1) [] {};
\node (ri) at (14.9,-3) [] {};

\node (w1) at (4.5,-1) [] {};
\node (wi) at (11,-3) [] {};

\draw (r1) node [above] {\normalsize {$R_0(X_1) \rightarrow v_{0_{1}}$}};
\draw (r2) node [above] {\normalsize {$R_2(X_1) \rightarrow v_{1_{1}}$}};
\draw (ri) node [above] {\normalsize {$R_{2i}(X_1) \rightarrow v_{i_{1}}$}};

\draw (w1) node [above] {\normalsize {$\forall X_{\ell}\in \mathcal{X}$: write $v_{1_{\ell}}$}}; 
\draw (w1) node [below] {\normalsize {$T_1$ commits}};

\draw (wi) node [above] {\normalsize {$\forall X_{\ell}\in \mathcal{X}$: write $v_{i_{\ell}}$}}; 
\draw (wi) node [below] {\normalsize {$T_{2i-1}$ commits}};

\begin{scope}   
\draw [|-,dotted] (0,0) node[left] {$\ms{Phase}~ 0|$ $T_0$} to (20,0);
\draw [|-|,thick] (0,0) node[left] {} to (2,0);
\end{scope}
\begin{scope}   
\draw [|-,dotted] (7,-1) node[left] {  $T_2$} to (20,-1);
\draw [|-|,thick] (3,-1) node[left] {$\ms{Phase}~ 1|$ ${T}_{1}$} to (6,-1);
\draw [|-|,thick] (7,-1) node[left] {} to (9,-1);
\end{scope}
\begin{scope}   
\draw [-,dotted] (10,-2) node[left] {} to (20,-2);
\end{scope}
\begin{scope}   
\draw [|-,dotted] (13.5,-3) node[left] {} to (20,-3);
\draw [|-|,thick] (9.5,-3) node[left] { $\ms{Phase}~ i|$ ${T}_{2i-1}$} to (12.5,-3);
\draw [|-|,thick] (13.5,-3) node[left] {$T_{2i}$} to (15.5,-3);
\end{scope}
\begin{scope}   
\draw [] (16,-4) node[above] {$\downarrow$} to (16,-4);
\draw [-,dotted] (15.5,-5) node[above] {extend to $c-1$ phases} to (18,-5);
\end{scope}
\end{tikzpicture}}}
	\\
	\vspace{2mm}
	\subfloat[extend every read-only transaction $T_{2i}$ in phase $i$ with t-reads of $X_2,\ldots X_{\ell}, \ldots $; each $\Read_{2i}(X_{\ell})$ must return $v_{i_{\ell}}$ \label{sfig:res-1}]{\scalebox{0.72}[0.72]{\begin{tikzpicture}
\node (r1) at (1.2,0) [] {};
\node (r2) at (8.4,-1) [] {};
\node (ri) at (14.9,-3) [] {};
\node (rf1) at (19.5,-0) [] {};
\node (rf2) at (19.5,-1) [] {};
\node (rfi) at (19.5,-3) [] {};

\node (w1) at (4.5,-1) [] {};
\node (wi) at (11,-3) [] {};

\draw (r1) node [above] {\normalsize {$R_0(X_1) \rightarrow v_{0_{1}}$}};
\draw (r2) node [above] {\normalsize {$R_2(X_1) \rightarrow v_{1_{1}}$}};
\draw (ri) node [above] {\normalsize {$R_{2i}(X_1) \rightarrow v_{i_{1}}$}};

\draw (rf1) node [above] {\normalsize {$R_0(X_2) \rightarrow v_{0_{2}} \cdots R_0(X_\ell)\rightarrow v_{0_{\ell}} \cdots$}};
\draw (rf2) node [above] {\normalsize {$R_2(X_2) \rightarrow v_{1_{2}} \cdots R_2(X_\ell)\rightarrow v_{1_{\ell}} \cdots$}};
\draw (rfi) node [above] {\normalsize {$R_{2i}(X_2) \rightarrow v_{i_{2}} \cdots R_{2i}(X_\ell)\rightarrow v_{i_{\ell}} \cdots$}};

\draw (w1) node [above] {\normalsize {$\forall X_{\ell} \in \mathcal{X}$: write $v_{1_{\ell}}$}}; 
\draw (w1) node [below] {\normalsize {$T_1$ commits}};

\draw (wi) node [above] {\normalsize {$\forall X_{\ell} \in \mathcal{X}$: write $v_{i_{\ell}}$}}; 
\draw (wi) node [below] {\normalsize {$T_{2i-1}$ commits}};

\begin{scope}   
\draw [|-,dotted] (0,0) node[left] {$T_0$} to (22,0);
\draw [|-|,thick] (0,0) node[left] {} to (2,0);
\draw [|-|,thick] (17,0) node[left] {} to (22,0);
\end{scope}
\begin{scope}   
\draw [|-,dotted] (7,-1) node[left] {$T_2$} to (22,-1);
\draw [|-|,thick] (3,-1) node[left] {${T}_{1}$} to (6,-1);
\draw [|-|,thick] (7,-1) node[left] {} to (9,-1);
\draw [|-|,thick] (17,-1) node[left] {} to (22,-1);
\end{scope}
\begin{scope}   
\draw [-,dotted] (10,-2) node[left] {} to (22,-2);
\end{scope}
\begin{scope}   
\draw [|-,dotted] (13.5,-3) node[left] {} to (22,-3);
\draw [|-|,thick] (9.5,-3) node[left] {${T}_{2i-1}$} to (12.5,-3);
\draw [|-|,thick] (13.5,-3) node[left] {$T_{2i}$} to (15.5,-3);
\draw [|-|,thick] (17,-3) node[left] {} to (22,-3);
\end{scope}
\begin{scope}   
\draw [] (16,-4) node[above] {$\downarrow$} to (16,-4);
\draw [-,dotted] (15.5,-5) node[above] {extend to $c-1$ phases} to (18,-5);
\end{scope}
\end{tikzpicture}}}
        
	\caption{Executions in the proof of Theorem~\ref{th:inv}; execution in \ref{sfig:res-0} must maintain $c$ distinct
	values of every t-object
        \label{fig:invdap}} 
\end{center}
\end{figure*}
\begin{definition}
Let $E$ be any execution of a TM implementation $M$.
We say that $E$ \emph{maintains $c$ distinct values $\{v_1,\ldots , v_c\}$ of t-object $X$},
if there exists an execution $E\cdot E'$ of $M$ such that
\begin{itemize}
\item
$E'$ contains the complete executions of $c$ t-reads of $X$ and,
\item
for all $i\in \{1,\ldots , c\}$, the response of the $i^{th}$ t-read of $X$ in $E'$ is $v_i$,
and if the response of the $i^{th}$ t-read of $X$ in $E'$ is $r\neq v_i$,
then $E\cdot E'$ is not an execution of $M$.
\end{itemize}
\end{definition}
\begin{theorem}
\label{th:inv}
Let $M$ be any TM implementation in $\mathcal{RWF}$ that uses invisible reads, and
$\mathcal{X}$, the set of all possible t-objects that may be accessed in any execution of $M$.
Then, for every $c \in \mathbb{N}$, there exists an execution $E$ of $M$
such that $E$ maintains at least $c$ distinct values of each t-object $X\in \mathcal{X}$.
\end{theorem}
\begin{proof}
Let $v_{0_{\ell}}$ be the initial value of t-object $X_{\ell}\in \mathcal{X}$.
For every $c\in \mathbb{N}$, we iteratively construct an execution $E$ of $M$ of the form depicted in Figure~\ref{sfig:res-0}.
The construction of $E$ proceeds in phases: there are at most $c-1$ phases. 
For all $i\in \{0,\ldots c-1\}$, we denote the execution after phase $i$ as $E_i$ which
is defined as follows:
\begin{itemize}
\item
$E_0$ is the complete step contention-free execution fragment $\alpha_0$ of read-only transaction $T_0$ that performs
$\Read_0(X_1)\rightarrow v_{0_{1}}$
\item
for all $i \in \{1,\ldots ,c-1\}$, $E_i$ is defined to be an execution of the form 
$\alpha_0\cdot \rho_1\cdot \alpha_1 \cdots \rho_i\cdot \alpha_i$ such that for all $j\in \{1,\ldots, i\}$,
\begin{itemize}
\item
$\rho_j$ is the t-complete step contention-free execution fragment of an updating transaction ${T}_{2j-1}$ that,
for all $X_{\ell} \in \mathcal{X}$ writes the value $v_{j_{\ell}}$ and commits
\item
$\alpha_j$ is the complete step contention-free execution fragment of a read-only transaction $T_{2j}$ that performs
$\Read_{2j}(X_1)\rightarrow v_{j_{1}}$
\end{itemize}
\end{itemize}
%
Since read-only transactions are invisible, for all $i\in \{0,\ldots, c-1\}$, the execution fragment $\alpha_i$ does not
contain any nontrivial events.
Consequently, for all $i<j\leq c-1$, the configuration after $E_i$ is indistinguishable to transaction ${T}_{2j-1}$
from a t-quiescent configuration and it must be committed in $\rho_{j}$ (by sequential progress for updating transactions).
Observe that, for all $1\leq j < i$, ${T}_{2j-1} \prec_{E}^{RT} {T}_{2i-1}$.
Strict serializability of $M$ now stipulates that, for all $i \in \{1,\ldots, c-1\}$, 
the t-read of $X_1$ performed by transaction $T_{2i}$ in the 
execution fragment $\alpha_{i}$ must return the value $v_{i_{1}}$ of $X_1$ as written by transaction ${T}_{2i-1}$ in
the execution fragment $\rho_i$ (in any serialization, 
${T}_{2i-1}$ is the latest committed transaction writing to $X_1$ that precedes $T_{2i}$). 
Thus, $M$ indeed has an execution $E$ of the form depicted in Figure~\ref{sfig:res-0}.

Consider the execution fragment $E'$ that extends $E$ in which, for all $i\in \{0,\ldots ,c-1 \}$, read-only transaction $T_{2i}$ 
is extended with the complete execution of the t-reads of every t-object 
$X_{\ell} \in \mathcal{X}\setminus \{X_1\}$ (depicted in Figure~\ref{sfig:res-1}).

We claim that, for all $i\in \{0,\ldots ,c-1 \}$, and for all $X_{\ell} \in \mathcal{X}\setminus \{X_1\}$,
$\Read_{2i}(X_{\ell})$ performed by transaction $T_{2i}$ must return the value $v_{i_{\ell}}$ of $X_{\ell}$ 
written by transaction ${T}_{2i-1}$ in the execution fragment $\rho_i$.
Indeed, by wait-free progress, $\Read_i(X_{\ell})$ must return a non-abort response in such an extension of $E$.
Suppose by contradiction that $\Read_i(X_{\ell})$ returns a response that is not $v_{i_{\ell}}$.
There are two cases: 
\begin{itemize}
\item
$\Read_{2i}(X_{\ell})$ returns the value $v_{j_{\ell}}$ written by transaction ${T}_{2j-1}$; $j<i$.
However, since for all $j < i$, ${T}_{2j} \prec_{E}^{RT} {T}_{2i}$, the execution is not strictly serializable---contradiction.
\item
$\Read_{2i}(X_{\ell})$ returns the value $v_{j_{\ell}}$ written by transaction ${T}_{2j}$; $j>i$.
Since $\Read_i(X_1)$ returns the value $v_{i_{1}}$ and ${T}_{2i} \prec_{E}^{RT} {T}_{2j}$,
there exists no such serialization---contradiction.
\end{itemize}
Thus, $E$ maintains at least $c$ distinct values of every t-object $X\in \mathcal{X}$.
\end{proof}
%
%
\section{Impossibility of strict disjoint-access parallelism}
\label{sec:lpdap}
In this section, we prove that it is impossible to derive TM implementations in $\mathcal{RWF}$ which ensure
that any two transactions accessing pairwise disjoint data sets can execute without contending on the same base object.
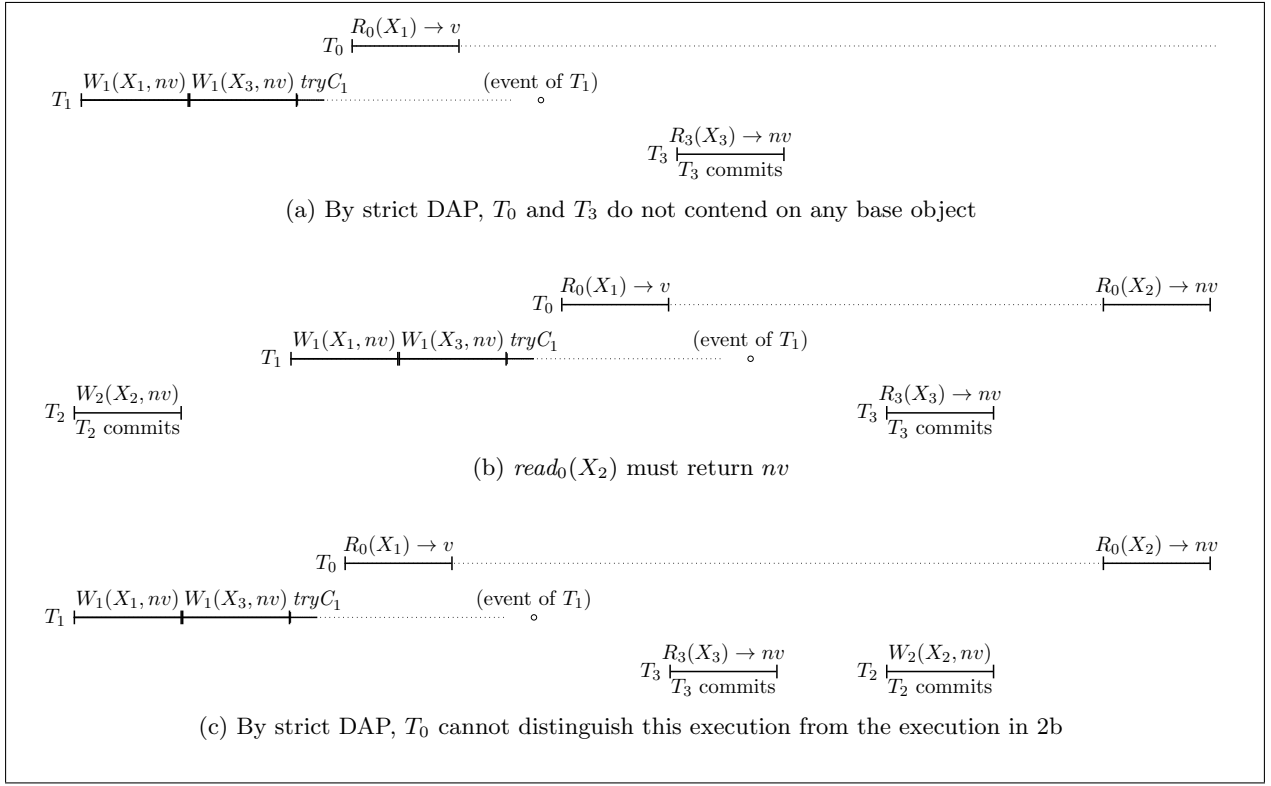
\begin{figure*}[t]
\begin{center}
	\subfloat[By strict DAP, $T_0$ and $T_3$ do not contend on any base object \label{sfig:dap-0}]{\scalebox{0.72}[0.72]{\begin{tikzpicture}
\node (r1) at (3,0) [] {};

\node (e) at (5.5,-1) [] {};

\node (w1) at (-2,-1) [] {};
\node (w2) at (0,-1) [] {};
\node (c) at (1.5,-1) [] {};

\node (r3) at (9,-2) [] {};

\draw (r1) node [above] {\normalsize{$R_0(X_1) \rightarrow v$}};

\draw (e) node [above] {\normalsize {(event of $T_1$)}};

\draw (w1) node [above] { \normalsize{$W_1(X_1,nv)$}};
\draw (w2) node [above] { \normalsize{$W_1(X_3,nv)$}};
\draw (c) node [above] { \normalsize{$\TryC_1$}};

\draw (r3) node [above] { \normalsize{$R_3(X_3) \rightarrow nv$}};
\draw (r3) node [below] { \normalsize{$T_3$ commits}};

\begin{scope}   
\draw [-,dotted] (2,0) node[left] {$T_0$} to (18,0);
\draw [|-|,thick] (2,0) node[left] {} to (4,0);
\end{scope}
\begin{scope}   
\draw [-,thick] (-3,-1) node[left] {$T_1$} to (1.5,-1);
\draw [|-|,thick] (-3,-1) node[left] {} to (-1,-1);
\draw [|-|,thick] (-1,-1) node[left] {} to (1,-1);
\draw [|-,dotted] (1,-1) node[left] {} to (1.5,-1);
\draw  (5.5,-1) circle [fill, radius=0.05]  (5.5,-1);
\draw [-,dotted] (-3,-1) to  (5,-1);
\end{scope}
\begin{scope}   
\draw [|-|,thick] (8,-2) node[left] {$T_3$} to (10,-2);
\end{scope}
\end{tikzpicture}}}
 	\\ 
	\vspace{2mm}
	\subfloat[$\Read_0(X_2)$ must return $nv$ \label{sfig:dap-1}]{\scalebox{0.72}[0.72]{\begin{tikzpicture}
\node (r1) at (3,0) [] {};
\node (r2) at (13,0) [] {};

\node (e) at (5.5,-1) [] {};

\node (w1) at (-2,-1) [] {};
\node (w2) at (0,-1) [] {};
\node (c) at (1.5,-1) [] {};

\node (r3) at (9,-2) [] {};
\node (w3) at (-6,-2) [] {};

\draw (r1) node [above] {\normalsize{$R_0(X_1) \rightarrow v$}};
\draw (r2) node [above] {\normalsize{$R_0(X_2)\rightarrow nv$}};

\draw (e) node [above] {\normalsize{(event of $T_1$)}};

\draw (w1) node [above] { \normalsize{$W_1(X_1,nv)$}};
\draw (w2) node [above] { \normalsize{$W_1(X_3,nv)$}};
\draw (c) node [above] { \normalsize{$\TryC_1$}};

\draw (r3) node [above] { \normalsize{$R_3(X_3) \rightarrow nv$}};
\draw (r3) node [below] { \normalsize{$T_3$ commits}};

\draw (w3) node [above] { \normalsize{$W_2(X_2,nv)$}};
\draw (w3) node [below] { \normalsize{$T_2$ commits}};

\begin{scope}   
\draw [-,dotted] (2,0) node[left] {$T_0$} to (14,0);
\draw [|-|,thick] (2,0) node[left] {} to (4,0);
\draw [|-|,thick] (12,0) node[left] {} to (14,0);
\end{scope}
\begin{scope}   
\draw [-,thick] (-3,-1) node[left] {$T_1$} to (1.5,-1);
\draw [|-|,thick] (-3,-1) node[left] {} to (-1,-1);
\draw [|-|,thick] (-1,-1) node[left] {} to (1,-1);
\draw [|-,dotted] (1,-1) node[left] {} to (1.5,-1);
\draw  (5.5,-1) circle [fill, radius=0.05]  (5.5,-1);
\draw [-,dotted] (-3,-1) to  (5,-1);
\end{scope}
\begin{scope}   
\draw [|-|,thick] (8,-2) node[left] {$T_3$} to (10,-2);
\draw [|-|,thick] (-7,-2) node[left] {$T_2$} to (-5,-2);
\end{scope}
\end{tikzpicture}}}
        \\ 
        \vspace{2mm}
	\subfloat[By strict DAP, $T_0$ cannot distinguish this execution from the execution in \ref{sfig:dap-1} \label{sfig:dap-2}]{\scalebox{0.72}[0.72]{\begin{tikzpicture}
\node (r1) at (3,0) [] {};
\node (r2) at (17,0) [] {};

\node (e) at (5.5,-1) [] {};

\node (w1) at (-2,-1) [] {};
\node (w2) at (0,-1) [] {};
\node (c) at (1.5,-1) [] {};

\node (r3) at (9,-2) [] {};
\node (w3) at (13,-2) [] {};

\draw (r1) node [above] {\normalsize{$R_0(X_1) \rightarrow v$}};
\draw (r2) node [above] {\normalsize{$R_0(X_2)\rightarrow nv$}};

\draw (e) node [above] {\normalsize {(event of $T_1$)}};

\draw (w1) node [above] { \normalsize{$W_1(X_1,nv)$}};
\draw (w2) node [above] { \normalsize{$W_1(X_3,nv)$}};
\draw (c) node [above] { \normalsize{$\TryC_1$}};

\draw (r3) node [above] { \normalsize{$R_3(X_3) \rightarrow nv$}};
\draw (r3) node [below] { \normalsize{$T_3$ commits}};

\draw (w3) node [above] { \normalsize{$W_2(X_2,nv)$}};
\draw (w3) node [below] { \normalsize{$T_2$ commits}};

\begin{scope}   
\draw [-,dotted] (2,0) node[left] {$T_0$} to (18,0);
\draw [|-|,thick] (2,0) node[left] {} to (4,0);
\draw [|-|,thick] (16,0) node[left] {} to (18,0);
\end{scope}
\begin{scope}   
\draw [-,thick] (-3,-1) node[left] {$T_1$} to (1.5,-1);
\draw [|-|,thick] (-3,-1) node[left] {} to (-1,-1);
\draw [|-|,thick] (-1,-1) node[left] {} to (1,-1);
\draw [|-,dotted] (1,-1) node[left] {} to (1.5,-1);
\draw  (5.5,-1) circle [fill, radius=0.05]  (5.5,-1);
\draw [-,dotted] (-3,-1) to  (5,-1);
\end{scope}
\begin{scope}   
\draw [|-|,thick] (8,-2) node[left] {$T_3$} to (10,-2);
\draw [|-|,thick] (12,-2) node[left] {$T_2$} to (14,-2);
\end{scope}
\end{tikzpicture}}}
	
	\caption{Executions in the proof of Theorem~\ref{th:lpdap}; execution in \ref{sfig:dap-2} is not strictly serializable
        \label{fig:indisdap}} 
\end{center}
\end{figure*}
\begin{theorem}
\label{th:lpdap}
There exists no strict DAP TM implementation in $\mathcal{RWF}$.
\end{theorem}
\begin{proof}
Suppose by contradiction that there exists a strict DAP TM implementation $M\in \mathcal{RWF}$.

Let $v$ be the initial value of t-objects $X_1$, $X_2$ and $X_3$.
Let $\pi$ be the t-complete step contention-free execution of transaction $T_1$ that
writes the value $nv \neq v$ to t-objects $X_1$ and $X_3$.
By sequential progress for updating transactions, $T_1$ must be committed in $\pi$.

Note that any read-only transaction that runs step contention-free after some prefix of $\pi$ must return a non-abort value.
Since any such transaction must return $v$ after the empty prefix of $\pi$ and $nv$ when it starts from $\pi$,
there exists $\pi'$, the longest prefix of $\pi$ that cannot be extended with the
t-complete step contention-free execution of any transaction that performs a t-read of $X_1$ and
returns $nv$ nor with the t-complete step contention-free execution of any transaction that performs a t-read of $X_3$
and returns $nv$.

Consider the execution fragment $\pi'\cdot \alpha_1$, where $\alpha_1$ is
the complete step contention-free execution of transaction $T_0$ that performs $\Read_0(X_1)\rightarrow v$.
Indeed, by definition of $\pi'$ and wait-free progress (assumed for read-only transactions),
$M$ has an execution of the form $\pi'\cdot \alpha_1$.

Let $e$ be the enabled event of transaction $T_1$ in the configuration after $\pi'$.
Without loss of generality, assume that $\pi'\cdot e$ can be extended with the t-complete
step contention-free execution of a transaction that reads $X_3$ and returns $nv$.

We now prove that $M$ has an execution of the form $\pi'\cdot \alpha_1 \cdot e \cdot \beta \cdot \gamma$, where
\begin{itemize}
\item
$\beta$ is the t-complete step contention-free execution fragment of transaction $T_3$ that performs $\Read_3(X_3)\rightarrow nv$
and commits
\item
$\gamma$ is the t-complete step contention-free execution fragment of transaction $T_2$ that
writes $nv$ to $X_2$ and commits.
\end{itemize}
Observe that, by definition of $\pi'$, $M$ has an execution of the form $\pi' \cdot e \cdot \beta$.
By construction, transaction $T_1$ applies a nontrivial primitive to a base object, say $b$ in the event $e$
that is accessed by transaction $T_3$ in the execution fragment $\beta$.
Since transactions $T_0$ and $T_3$ access mutually disjoint data sets in $\pi'\cdot \alpha_1 \cdot e \cdot \beta$,
$T_3$ does not access any base object in $\beta$ to which transaction $T_0$ applies a nontrivial primtive
in the execution fragment $\alpha_1$ (assumption of strict DAP).
Thus, $\alpha_1$ does not contain a nontrivial primitive to $b$ and $\pi'\cdot \alpha_1 \cdot e \cdot \beta$
is indistinguishable to $T_3$ from the execution $\pi' \cdot e \cdot \beta$.
This proves that $M$ has an execution of the form $\pi'\cdot \alpha_1 \cdot e \cdot \beta$ (depicted in Figure~\ref{sfig:dap-0}).

Since transaction $T_2$ writes to t-object $\Dset(T_2)=X_2 \not\in \{\Dset(T_1)\cup \Dset(T_0)\cup \Dset(T_3)\}$,
by strict DAP,
the configuration after $\pi'\cdot \alpha_1 \cdot e \cdot \beta$ is indistinguishable to $T_2$
from a t-quiescent configuration. 
Indeed, transaction $T_2$ does not contend with
any of the transactions $T_1$, $T_0$ and $T_3$ on any base object in $\pi'\cdot \alpha_1 \cdot e \cdot \beta \cdot \gamma$.
Sequential progress of $M$ requires that
$T_2$ must be committed in $\pi'\cdot \alpha_1 \cdot e \cdot \beta \cdot \gamma$.
Thus, $M$ has an execution of the form $\pi'\cdot \alpha_1 \cdot e \cdot \beta \cdot \gamma$.

By the above arguments, the execution $\pi'\cdot \alpha_1 \cdot e \cdot \beta \cdot \gamma$
is indistinguishable to each of the transactions $T_1$, $T_0$, $T_2$ and $T_3$ from
$\gamma \cdot \pi'\cdot \alpha_1 \cdot e \cdot \beta$ in which transaction $T_2$ precedes $T_1$ in real-time ordering.
Thus, $\gamma \cdot \pi'\cdot \alpha_1 \cdot e \cdot \beta$ is also an execution of $M$.

Consider the extension of the execution $\gamma \cdot \pi'\cdot \alpha_1 \cdot e \cdot \beta$
in which transaction $T_0$ performs $\Read_0(X_2)$ and commits (depicted in Figure~\ref{sfig:dap-1}). 
Strict serializability of $M$ stipulates that $\Read_0(X_2)$ must return $nv$
since $T_2$ (which writes $nv$ to $X_2$ in $\gamma$) precedes $T_0$ in this execution.

Similarly, we now extend the execution $\pi'\cdot \alpha_1 \cdot e \cdot \beta \cdot \gamma$
with the complete step contention-free execution fragment of the t-read of $X_2$ by transaction $T_0$. 
Since $T_0$ is a read-only transaction,
it must be committed in this extension.
However, as proved above, this execution is indistinguishable to $T_0$ from the execution
depicted in Figure~\ref{sfig:dap-1} in which $\Read_0(X_2)$ must return $nv$.
Thus, $M$ has an execution of the form $\pi'\cdot \alpha_1 \cdot e \cdot \beta \cdot \gamma \cdot \alpha_2$,
where $T_0$ performs $\Read_0(X_2) \rightarrow nv$ in $\alpha_2$ and commits.

However, the execution $\pi'\cdot \alpha_1 \cdot e \cdot \beta \cdot \gamma \cdot \alpha_2$ (depicted in Figure~\ref{sfig:dap-2})
is not strictly serializable.
Indeed, transaction $T_1$ must be committed in any serialization and must precede transaction $T_3$
since $\Read_3(X_3)$ returns the value of $X_3$ written by $T_m$. However, transaction $T_0$ must
must precede $T_1$ since $\Read_0(X_1)$ returns the initial the value of $X_1$.
Also, transaction $T_2$ must precede $T_0$ since $\Read_0(X_2)$ returns the value of $X_2$ written by $T_2$.
But transaction $T_3$ must precede $T_2$ to respect real-time ordering of transactions.
Thus, $T_1$ must precede $T_0$ in any serialization. But there exists no such serialization: a contradiction to the
assumption that $M$ is strictly serializable.
\end{proof}
%
%
\section{A linear lower bound on expensive synchronization}
\label{sec:rwl}
Attiya \emph{et al.} identified two common expensive synchronization patterns that frequently arise in
the design of concurrent algorithms: \emph{read-after-write (RAW) or atomic write-after-read (AWAR)}~\cite{AGK11-popl,McKenney10}.
In this section, we prove a linear lower bound (in the size of the transaction's read set) on the number of RAWs or AWARs
for weak DAP TM implementations in $\mathcal{RWF}$.
To so so, we construct an execution in which each t-read
operation of an arbitrarily long read-only transaction contains a RAW or an AWAR. 
\begin{definition}
(RAW/AWAR metric) Let $\pi$ be a fragment of an execution of a TM implementation $M$ and let $\pi^i$ denote the $i$-th
event in $\pi$ ($i=0,\ldots , |\pi|-1$).   

We say that a transaction $T$ performs a \emph{RAW} (read-after-write) in $\pi$ if  
$\exists i,j; 0 \leq i < j < |\pi|$
such that (1) $\pi^i$ is a write to a base object $b$ by $T$, 
(2) $\pi^j$ is a read of a base object $b'\neq b$ by $T$ and 
(3) there is no $\pi^k$ such that $i<k<j$ and $\pi^k$ is a write to $b'$ by $T$.

We say a transaction $T$ performs an \emph{AWAR} (atomic-write-after-read)
in $\pi$ if $\exists i, 0 \leq i < |\pi|$ such that the event $\pi^i$
is the application of a nontrivial rmw primitive that reads a base object $b$ followed by a write to $b$.
\end{definition}
%
\begin{figure*}[t]
\begin{center}
	\subfloat[$\Read_0(X_j)\rightarrow v$ performs no RAW/AWAR; $T_0$ and $T_j$ are unaware of step contention \label{sfig:inv-1}]{\scalebox{0.72}[0.72]{\begin{tikzpicture}
\node (r1) at (3,0) [] {};
\node (r2) at (7.7,0) [] {};

\node (w1) at (7.5,-2) [] {};

\draw (r1) node [below] {\normalsize {$R_0(X_1) \cdots R_0(X_{j-1})$}};
\draw (r1) node [above] {\normalsize {$j-1$ t-reads}};

\draw (r2) node [above] {\normalsize {$R_0(X_j)\rightarrow v$}};
\draw (r2) node [below] {\normalsize {initial value}};

\draw (w1) node [above] {\normalsize {$W_j(X_j,nv)$}}; 
\draw (w1) node [below] {\normalsize {$T_j$ commits}};

\begin{scope}   
\draw [|-|,thick] (0,0) node[left] {$T_0$} to (6,0);
\draw [|-|,thick] (6.5,0) node[left] {} to (9,0);
\draw [-,dotted] (0,0) node[left] {} to (9,0);
\end{scope}
\begin{scope}   
\draw [|-|,thick] (6.5,-2) node[left] {$T_j$} to (9,-2);
\end{scope}
\end{tikzpicture}}}
        \\
        \vspace{2mm}
	\subfloat[$R_0(X_{m})$ must return $nv$ by strict serializability \label{sfig:inv-2}]{\scalebox{0.72}[0.72]{\begin{tikzpicture}
\node (r1) at (3,0) [] {};
\node (r2) at (7.7,0) [] {};
\node (r3) at (12.5,0) [] {};

\node (w1) at (7.5,-2) [] {};

\node (w2) at (-3,-2) [] {};

\draw (r1) node [below] {\small {$R_0(X_1) \cdots R_0(X_{j-1})$}};
\draw (r1) node [above] {\small {$j-1$ t-reads}};

\draw (r2) node [above] {\small {$R_0(X_j)\rightarrow v$}};
\draw (r2) node [below] {\small {initial value}};

\draw (w1) node [above] {\small {$W_j(X_j,nv)$}}; 
\draw (w1) node [below] {\small {$T_j$ commits}};

\draw (w2) node [above] {\small {$W_{\ell}(X_{m},nv)$}}; 
\draw (w2) node [below] {\small {$T_{\ell}$ commits}};

\draw (r3) node [above] {\small {$R_0(X_{m})\rightarrow nv$}};
\draw (r3) node [below] {\small {new value}};

\begin{scope}   
\draw [|-|,thick] (0,0) node[left] {$T_0$} to (6,0);
\draw [|-|,thick] (6.5,0) node[left] {} to (9,0);
\draw [|-|,thick] (-4,-2) node[left] {$T_{\ell}$} to (-1,-2);
\draw [|-|,thick] (11,0) node[left] {} to (13.5,0);
\draw [-,dotted] (0,0) to (13.5,0);
\end{scope}
\begin{scope}   
\draw [|-|,thick] (6.5,-2) node[left] {$T_j$} to (9,-2);
\end{scope}
\end{tikzpicture}}}
	\\ 
	\vspace{2mm}
	\subfloat[By weak DAP, $T_0$ cannot distinguish this execution from \ref{sfig:inv-2} \label{sfig:inv-3}]{\scalebox{0.72}[0.72]{\begin{tikzpicture}
\node (r1) at (3,0) [] {};
\node (r2) at (7.7,0) [] {};
\node (r3) at (17.5,0) [] {};

\node (w1) at (7.7,-2) [] {};

\node (w2) at (12.5,-2) [] {};

\draw (r1) node [below] {\normalsize {$R_0(X_1) \cdots R_0(X_{j-1})$}};
\draw (r1) node [above] {\normalsize {$j-1$ t-reads}};

\draw (r2) node [above] {\normalsize {$R_0(X_j)\rightarrow v$}};
\draw (r2) node [below] {\normalsize {initial value}};

\draw (w1) node [above] {\normalsize {$W_j(X_j,nv)$}}; 
\draw (w1) node [below] {\normalsize {$T_j$ commits}};

\draw (w2) node [above] {\normalsize {$W_{{\ell}}(X_{m},nv)$}}; 
\draw (w2) node [below] {\normalsize {$T_{\ell}$ commits}};

\draw (r3) node [above] {\normalsize {$R_0(X_{m})\rightarrow nv$}};
\draw (r3) node [below] {\normalsize {new value}};

\begin{scope}   
\draw [|-|,thick] (0,0) node[left] {$T_0$} to (6,0);
\draw [|-|,thick] (6.5,0) node[left] {} to (9,0);
\draw [|-|,thick] (11,-2) node[left] {$T_{\ell}$} to (14,-2);
\draw [|-|,thick] (16,0) node[left] {} to (18.5,0);
\draw [-,dotted] (0,0) to (18.5,0);
\end{scope}
\begin{scope}   
\draw [|-|,thick] (6.5,-2) node[left] {$T_j$} to (9,-2);
\end{scope}
\end{tikzpicture}}}
	\caption{Executions in the proof of Theorem~\ref{th:rwf}; execution in \ref{sfig:inv-3} is not strictly serializable
        \label{fig:indis}} 
\end{center}
\end{figure*}
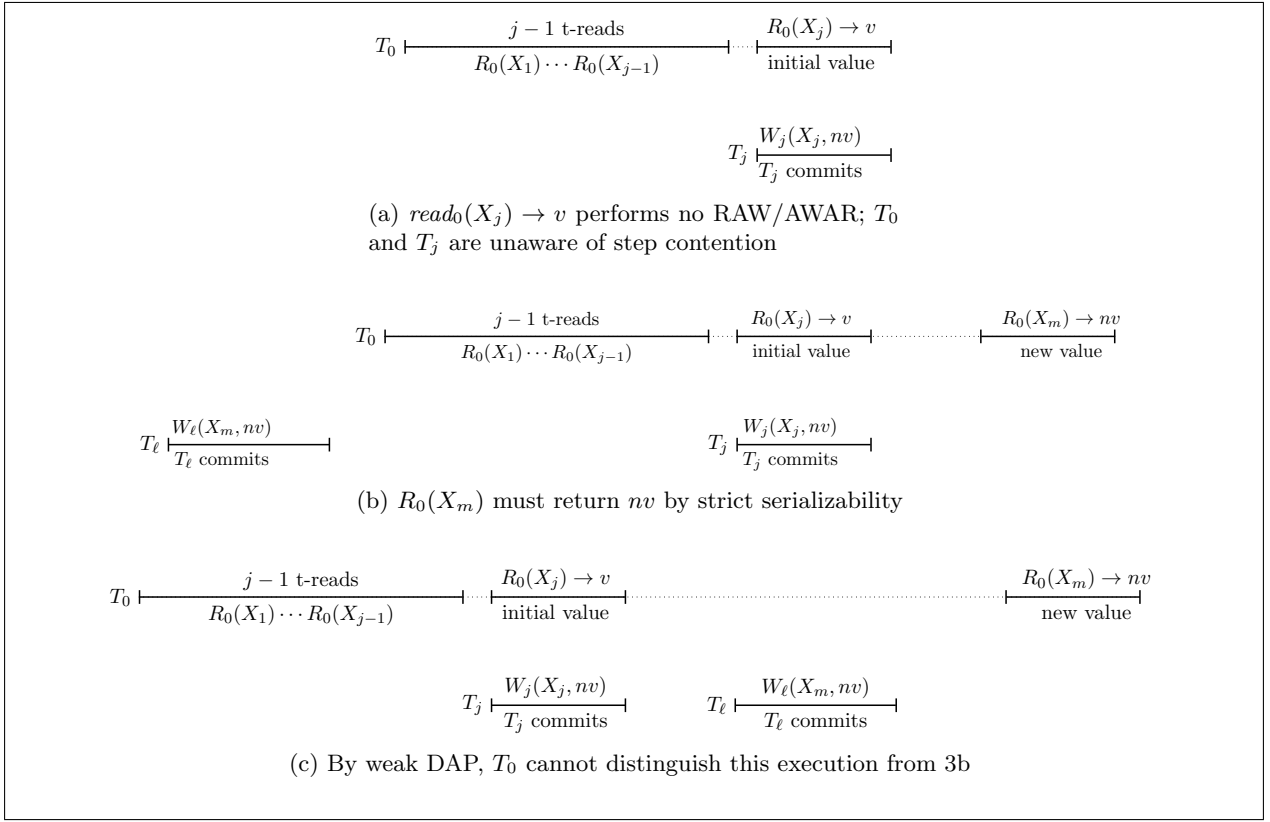
\begin{theorem}
\label{th:rwf}
Every weakly DAP TM implementation $M \in \mathcal{RWF}$
has, for all $m\in \mathbb{N}$, an
execution in which some read-only transaction $T_0$ 
with $m=|\Rset(T_0)|$ performs $\Omega(m)$ RAWs/AWARs.
\end{theorem}
\begin{proof}
Let $v$ be the initial value of each of the t-objects $X_1,\ldots , X_m$.
Consider the t-complete step contention-free execution
of transaction $T_0$ that performs $m$ t-reads $\Read_0(X_1)$, $\Read_0(X_1)$,$\ldots \Read_0(X_m)$ and commits.
We prove that each of the first $m-1$ t-reads must perform a RAW or an AWAR.

For each $j\in \{1,\ldots , m-1\}$, we define an execution of the form $\alpha_1\cdot \alpha_2 \cdots \alpha_j$, 
where for all $i \in \{1,\ldots , j\}$, $\alpha_i$ is the complete step contention-free execution fragment
of $\Read_0(X_j) \rightarrow v$.
Suppose by contradiction that $\alpha_j$ does not contain a RAW or an AWAR.

The following claim shows that we can schedule a committed transaction $T_j$ that writes a new value to $X_j$ 
concurrent to $\Read_0(X_j)$ such that
the execution is indistinguishable to both $T_0$ and $T_j$ from a 
step contention-free execution (depicted in Figure~\ref{sfig:inv-1}).
\begin{claim}
\label{cl:one}
For all $j\in \{1,\ldots , m-1\}$, $M$ has an execution of the form 
$\alpha_1 \cdots \alpha_{j-1} \cdot \alpha^1_j \cdot \delta_j \cdot \alpha^2_j$
where,
\begin{itemize}
\item
$\delta_{j}$ is the t-complete step contention-free execution fragment of transaction $T_{j}$ that
writes $nv \neq v$ and commits
\item
$\alpha^1_j\cdot \alpha^2_j=\alpha_j$ is the complete execution fragment of the $j^{th}$ t-read $\Read_0(X_j) \rightarrow v$
such that
\begin{itemize}
\item
$\alpha^1_j$ does not contain any nontrivial events
\item
$\alpha_1 \cdots \alpha_{j-1} \cdot \alpha^1_j \cdot \delta_j \cdot \alpha^2_j$ is indistinguishable 
to $T_0$ from the step contention-free
execution fragment $\alpha_1 \cdots \alpha_{j-1} \cdot \alpha^1_j \cdot \alpha^2_j$
\end{itemize}
\end{itemize}
Moreover, $T_j$ does not access any base object to which $T_0$ applies a nontrivial event in
$\alpha_1 \cdots \alpha_{j-1} \cdot \alpha^1_j \cdot \delta_j$.
\end{claim}
\begin{proof}
By wait-free progress (for read-only transactions) and strict serializability, $M$ has an execution of the form 
$\alpha_1 \cdots \alpha_{j-1}$ in which each of the t-reads performed by $T_0$ must return the initial value of the t-objects.

Since $T_j$ is an updating transaction, by sequential progress, there exists an execution of $M$
of the form $\delta_j\cdot \alpha_1 \cdots \alpha_{j-1}$.
Since $T_0$ and $T_j$ are disjoint-access in the $\delta_j \cdot \alpha_1 \cdots \alpha_{j-1}$, by Lemma~\ref{lm:dap},
$T_0$ and $T_j$ do not contend on any base object in $\delta_j \cdot \alpha_1 \cdots \alpha_{j-1}$.
Thus, $\alpha_1 \cdots \alpha_{j-1}  \cdot \delta_j$ is indistinguishable to $T_j$ from the execution $\delta_j$
and $\alpha_1 \cdots \alpha_{j-1}  \cdot \delta_j$ is also an execution of $M$.

Let $e$ be the first event that contains a write to a base object in $\alpha_j$.
If there exists no such write event to a base object in $\alpha_j$, then $\alpha^1_j= \alpha_j$
and $\alpha^2_j$ is empty.
Otherwise, we represent the execution fragment $\alpha_j$ as $\alpha^1_{j}\cdot e \cdot \alpha^f_{j}$.

Since $\alpha^s_{j}$ does not contain any nontrivial events that write to a base object, 
$\alpha_1 \cdots \alpha_{j-1} \cdot \alpha^s_j \cdot \delta_j$ is indistinguishable to transaction $T_j$
from the execution $\alpha_1 \cdots \alpha_{j-1} \cdot \delta_j$.
Thus, $\alpha_1 \cdots \alpha_{j-1} \cdot \alpha^s_j \cdot \delta_j$ is an execution of $M$.
Since $e$ is not an atomic-write-after-read, $\alpha_1 \cdots \alpha_{j-1} \cdot \alpha^s_j \cdot \delta_j\cdot e$
is an execution of $M$.
Since $\alpha_j$ does not contain a RAW, any read performed in $\alpha^f_{j}$ may only be performed
to base objects previously written in $e \cdot \alpha^f_{j}$.
Thus, $\alpha_1 \cdots \alpha_{j-1} \cdot \alpha^s_j \cdot \delta_j \cdot e\cdot \alpha^f_j$
is indistinguishable to transaction $T_0$ from the step contention-free execution
$\alpha_1 \cdots \alpha_{j-1} \cdot \alpha^s_j \cdot e\cdot \alpha^f_j$ in which $\Read_0(X_j) \rightarrow v$.

Choosing $\alpha^2_j= e\cdot \alpha^f_j$,
it follows that $M$ has an execution of the form 
$\alpha_1 \cdots \alpha_{j-1} \cdot \alpha^1_j \cdot \delta_j \cdot \alpha^2_j$
that is indistinguishable to $T_j$ and $T_0$ from a step contention-free execution.
The proof follows. 
\end{proof}
We now prove that, for all $j\in \{1,\ldots , m-1\}$, 
$M$ has an execution of the form 
$\delta_{m}\cdot \alpha_1 \cdots \alpha_{j-1} \cdot \alpha^1_j \cdot \delta_j \cdot \alpha^2_j$
such that
\begin{itemize}
\item
$\delta_{m}$ is the t-complete step contention-free execution of transaction $T_{\ell}$
that writes $nv\neq v$ to $X_m$ and commits
\item
$T_{\ell}$ and $T_0$ do not contend on any base object in 
$\delta_{m}\cdot \alpha_1 \cdots \alpha_{j-1} \cdot \alpha^1_j \cdot \delta_j \cdot \alpha^2_j$
\item
$T_{\ell}$ and $T_j$ do not contend on any base object in
$\delta_{m}\cdot \alpha_1 \cdots \alpha_{j-1} \cdot \alpha^1_j \cdot \delta_j \cdot \alpha^2_j$.
\end{itemize}

By sequential progress for updating transactions, $T_{\ell}$ which writes the value $nv$ to $X_{m}$ must be committed in
$\delta_{m}$ since it is running in the absence of step-contention from the initial configuration.
Observe that $T_{\ell}$ and $T_0$ are disjoint-access in 
$\delta_{m}\cdot \alpha_1 \cdots \alpha_{j-1} \cdot \alpha^1_j\cdot \delta_j \cdot \alpha^2_j$.
By definition of $\alpha^1_j$ and $\alpha^2_j$, 
$\delta_{m}\cdot \alpha_1 \cdots \alpha_{j-1} \cdot \alpha^1_j \cdot \delta_j \cdot \alpha^2_j$
is indistinguishable to $T_0$ from $\delta_{m} \cdot \alpha_1 \cdots \alpha_{j-1} \cdot \alpha^1_j \cdot \alpha^2_j$.
By Lemma~\ref{lm:dap}, $T_{\ell}$ and $T_0$ do not contend on any base object in 
$\delta_{m} \cdot \alpha_1 \cdots \alpha_{j-1} \cdot \alpha^1_j \cdot \alpha^2_j$.

By Claim~\ref{cl:one}, $\delta_{m}\cdot \alpha_1 \cdots \alpha_{j-1} \cdot \alpha^1_j \cdot \delta_j$
is indistinguishable to $T_j$ from $\delta_{m}\cdot \delta_j$.
But transactions $T_{\ell}$ and $T_j$ are disjoint-access in $\delta_{m}\cdot \delta_j$, and by Lemma~\ref{lm:dap},
$T_j$ and $T_{\ell}$ do not contend on any base object in $\delta_{m}\cdot \delta_j$.

Since strict serializability of $M$ stipulates that each of the $j$ t-reads performed by $T_0$ return the initial
values of the respective t-objects, $M$ has an execution of the form 
$\delta_{m}\cdot \alpha_1 \cdots \alpha_{j-1} \cdot \alpha^1_j \cdot \delta_j \cdot \alpha^2_j$.

Consider the extension of $\delta_{m}\cdot \alpha_1 \cdots \alpha_{j-1} \cdot \alpha^1_j \cdot \delta_j \cdot \alpha^2_j$
in which $T_0$ performs $(m-j)$ t-reads
of $X_{j+1},\cdots , X_m$ step contention-free and commits (depicted in Figure~\ref{sfig:inv-2}). 
By wait-free progress of $M$ and since $T_0$ is a read-only transaction, there exists such an execution.
Notice that the $m^{th}$ t-read, $\Read_0(X_{m})$
must return the value $nv$ by strict serializability since $T_{\ell}$ precedes $T_0$ in real-time order in this execution.

Recall that neither pairs of transactions $T_{\ell}$ and $T_{j}$ nor $T_{\ell}$ and $T_0$ contend on any base object in
the execution $\delta_{m}\cdot \alpha_1 \cdots \alpha_{j-1} \cdot \alpha^1_j \cdot \delta_j \cdot \alpha^2_j$.
It follows that for all $j\in \{1,\ldots , m-1\}$, 
$M$ has an execution of the form
$\alpha_1 \cdots \alpha_{j-1} \cdot \alpha^1_j \cdot \delta_j \cdot \alpha^2_j\cdot \delta_{m}$
in which $T_{j}$ precedes $T_{\ell}$ in real-time order.

Let $\alpha'$ be the execution fragment that extends
$ \alpha_1 \cdots \alpha_{j-1} \cdot \alpha^1_j \cdot \delta_j \cdot \alpha^2_j\cdot \delta_{m}$
in which $T_0$ performs $(m-j)$ t-reads
of $X_{j+1},\cdots , X_m$ step contention-free and commits (depicted in Figure~\ref{sfig:inv-3}).
Since
$\alpha_1 \cdots \alpha_{j-1} \cdot \alpha^1_j \cdot \delta_j \cdot \alpha^2_j \cdot \delta_{m}$
is indistinguishable to $T_0$ from the execution
$\delta_{m}\cdot \alpha_1 \cdots \alpha!_{j-1} \cdot \alpha^1_j \cdot \delta_j \cdot \alpha^2_j$,
$\Read_0(X_{m})$ must return the response value $nv$ in $\alpha'$. 

The execution 
$ \alpha_1 \cdots \alpha_{j-1} \cdot \alpha^1_j \cdot \delta_j \cdot \alpha^2_j\cdot \delta_{m} \cdot \alpha'$
is not strictly serializable. 
In any serialization, $T_j$ must precede $T_{\ell}$ to respect the real-time ordering of transactions,
while $T_{\ell}$ must precede $T_0$ since $\Read_j(X_{m})$ returns the value of $X_{m}$ updated by $T_{\ell}$. Also,
transaction $T_0$ must precede $T_j$ since $\Read_0(X_j)$ returns the initial value of $X_j$. 
But there exists no such serialization: a contradiction to the assumption that $M$ is strict serializable.

Thus, for all $j\in \{1,\ldots , m-1\}$, transaction $T_0$ must perform a RAW or an AWAR during the execution of $\Read_0(X_j)$,
completing the proof.
\end{proof}
Since Theorem~\ref{th:rwf} implies that read-only transactions must perform nontrivial events, we have the following corollary
that was proved directly in \cite{AHM09}.
\begin{corollary}[\cite{AHM09}]
\label{cr:inv}
There does not exist any weak DAP TM implementation $M\in \mathcal{RWF}$
that uses invisible reads.
\end{corollary}
\section{Related work}
\label{sec:related}
Strict DAP was introduced by Guerraoui and Kapalka~\cite{tm-book} who proved that it is impossible to implement
\emph{obstruction-free} TMs (transactions running in the absence of step contention must commit) that satisfy strict DAP. 
$\mathcal{RWF}$ is incomparable to the class of obstruction-free TMs,
as is the proof technique used to establish the impossibility.

Attiya \emph{et al.}~\cite{AHM09} introduced the notion of weak DAP 
and showed that it is impossible to implement 
weak DAP strictly serializable TMs in $\mathcal{RWF}$ 
if read-only transactions may only apply trivial primitives to base objects.
Attiya et al.~\cite{AHM09} also considered a stronger ``disjoint-access'' property, 
called simply DAP, referring to the
original definition proposed Israeli and
Rappoport~\cite{israeli-disjoint}.
In DAP,  two transactions are allowed to \emph{concurrently access} (even
for reading) the same base object only if they are disjoint-access.
For an $n$-process DAP TM implementation,  \cite{AHM09} showed that a read-only transaction must
perform at least $n-3$ writes. 
Our lower bound is strictly stronger than the one in~\cite{AHM09}, as
it assumes only weak DAP,
considers a more precise RAW/AWAR metric, and does not depend on the
number of processes in the system.
(Technically, the last point follows from the fact that the execution
constructed in the proof of Theorem~\ref{th:rwf}
uses only  $3$ concurrent processes.)
Thus, 
the theorem strengthens and subsumes the two lower bounds of~\cite{AHM09} within a single proof.

Perelman \emph{et al.}, considered the class of 
\emph{mv-permissive} TMs:
only updating transactions that read-write conflict on data items with another
updating transaction may be aborted~\cite{PFK10}.
In general, $\mathcal{RWF}$ is incomparable
with the class of \emph{mv-permissive} TMs. On the one hand, as we observed,
mv-permissiveness does not guarantee that read-only transactions are not guaranteed to commit
in a wait-free manner.  
On the other hand, $\mathcal{RWF}$ allows an
updating transaction to abort in the presence of a concurrent
read-only transaction, which is disallowed by mv-permissive TMs.       
Assuming that every t-operation returns in a finite number of its own
steps (possibly with an \emph{abort}, we call this wait-free TM-liveness), \cite{PFK10} showed that 
implementing a weak DAP mv-permissive TM is impossible.
While $\mathcal{RWF}$ is not subject to this
impossibility, it is impossible to implement strict DAP implementations in $\mathcal{RWF}$.
In principle, the class of TMs considered in~\cite{PFK10} is a proper subset of our
$\mathcal{RWF}$.
 
Also, \cite{PFK10} proved that mv-permissive TMs cannot be \emph{online space optimal} i.e. no mv-permissive TM can be optimal
in terms of number of versions kept. 
Our result on the space complexity of implementations in $\mathcal{RWF}$ that use invisible reads
is strictly stronger since it applies to a larger class of TMs and 
proves that such implementations must maintain unbounded number of versions
of every data item. 
Moreover, it is easy to see that our proof is also applicable for the class of TMs considered in~\cite{PFK10}. 

Attiya \emph{et al.} introduced the RAW/AWAR metric and proved that it is impossible to derive RAW/AWAR-free implementations of
a wide class of data types.
The metric has been used in~\cite{KR11}
to measure the complexity of read-only transactions in a strictly
stronger class of \emph{permissive} TMs that provide wait-free TM-liveness, in which a transaction
may be aborted only if committing it would violate opacity. 
Detailed coverage on memory fences and the RAW/AWAR metric can be found in \cite{McKenney10}.
%
%
%
\section{Concluding remarks}
\label{sec:disc}
In this paper, we studied the issue of providing different progress 
guarantees to different classes of transactions, assuming that
read-only transactions are wait-free, 
but updating transactions are guaranteed to commit only when they run sequentially.
First, we prove that if read-only transactions are required to be invisible, then any strictly serializable TM implementation
with these progress guarantees
must maintain an unbounded number of values for every data item.
Second, we prove that strictly serializable TMs with these progress guarantees cannot be disjoint-access parallel 
in a strict sense.
Then, assuming a weaker form of disjoint-access-parallelism, we show that
there exists an arbitrarily long read-only transaction that performs an expensive
synchronization pattern within each of its read operations. 
Our lower bounds also hold for stronger TM-correctness conditions like
\emph{opacity}~\cite{tm-book}, \emph{virtual-world consistency}~\cite{damien-vw}, \emph{TMS1} and \emph{TMS2}~\cite{TMS}.

Some questions remain open. Is the lower bound of Theorem~\ref{th:rwf} tight? We conjecture that it is not.
Can we establish a fundamental tradeoff between the complexity of read-only and updating transactions
incurred by implementations in $\mathcal{RWF}$? More generally,
assuming transactional operations provide wait-free termination, what kind of
transactions can be provided with unconditional progress?
Addressing these questions is ongoing and future work.
%
%
%
\bibliography{references}
%
\end{document}